\documentclass[11pt]{article}

\usepackage{amsmath,amsfonts,amssymb,amsthm}
\usepackage[margin=1in]{geometry}
\usepackage[colorlinks]{hyperref}
\usepackage{enumitem}
\usepackage{mathtools}
\usepackage{enumitem}
\usepackage{algorithm, caption}
\usepackage[noend]{algpseudocode}
\usepackage{subfigure}
\usepackage{graphicx}
\usepackage{tabularx,environ}
\usepackage[braket]{qcircuit}
\usepackage[dvipsnames]{xcolor}
\usepackage{multirow}
\usepackage[capitalize]{cleveref}
\usepackage{dsfont}
\usepackage{enumitem}
\usepackage{comment}

\makeatletter
\newcommand*\bigcdot{\mathpalette\bigcdot@{.5}}
\newcommand*\bigcdot@[2]{\mathbin{\vcenter{\hbox{\scalebox{#2}{$\m@th#1\bullet$}}}}}
\makeatother

    \hypersetup{
    	colorlinks,
    	linkcolor={red!100!black},
    	citecolor={blue!100!black},
    }
\makeatletter
\newcommand\notsotiny{\@setfontsize\notsotiny\@vipt\@viipt}
\makeatother

\newcommand{\C}{{\mathbb{C}}}

\newcommand{\norm}[1]{\left\|{#1}\right\|}

\def\01{\{0,1\}}

\newcommand{\eps}{\varepsilon}

\theoremstyle{plain}

\newtheorem{theorem}{Theorem}

\newtheorem{lemma}[theorem]{Lemma}

\newtheorem{fact}[theorem]{Fact}
\newtheorem{result}[theorem]{Result}


\definecolor{applegreen}{rgb}{0.0, 0.5, 0.0}

\DeclareMathOperator{\op}{op} 

\newcommand{\poly}{\mbox{\rm poly}}

\DeclareMathOperator{\Tr}{Tr}
\DeclareMathOperator{\tr}{tr}

\newcommand{\beq}{\begin{equation}}
\newcommand{\eeq}{\end{equation}}
\newcommand{\beqn}{\begin{equation*}}
\newcommand{\eeqn}{\end{equation*}}
\newcommand{\beqr}{\begin{eqnarray}}
\newcommand{\eeqr}{\end{eqnarray}}
\newcommand{\beqrn}{\begin{eqnarray*}}
\newcommand{\eeqrn}{\end{eqnarray*}}
\newcommand{\bmline}{\begin{multline}}
\newcommand{\emline}{\end{multline}}
\newcommand{\bmlinen}{\begin{multline*}}
\newcommand{\emlinen}{\end{multline*}}

\theoremstyle{plain}

\theoremstyle{definition}

\theoremstyle{remark}
\newtheorem{remark}[theorem]{Remark}

\renewenvironment{proof}[1][]{
    \begin{trivlist}
    \item[\hspace{\labelsep}{\em\noindent Proof#1:\/}]}
    {{\hfill$\Box$}
    \end{trivlist}
}

\newtheoremstyle{named}{}{}{\itshape}{}{\bfseries}{.}{.5em}{\thmnote{#3}}
\theoremstyle{named}

\usepackage{tikz}
\usepackage{thm-restate} %To repeat theorems, lemmas etc....

\author{
 Andreas Bluhm\\
  \texttt{Univ.\ Grenoble Alpes, CNRS}\\\texttt{Grenoble INP, LIG}\\ \texttt{ 38000 Grenoble, France}
  \and
  Matthias C. Caro\\
  \texttt{University of Warwick}\\
  \texttt{Coventry CV4 7AL, United Kingdom}
  \and 
  Francisco Escudero Gutiérrez\\
  \texttt{Centrum Wiskunde \& Informatica, CWI}\\
  \texttt{1098 XG, Amsterdam, The Netherlands}
  \and
  Aadil Oufkir\\
  \texttt{Institute for Quantum Information}\\
  \texttt{
  RWTH Aachen University}\\\texttt{
  Aachen, Germany}
  \and 
  Cambyse Rouzé\\
  \texttt{Inria, Télécom Paris - LTCI}\\ 
\texttt{  Institut Polytechnique de Paris}\\
\texttt{91120 Palaiseau, France}
}

\begin{document}

\title{Certifying and learning quantum Ising Hamiltonians}

\date{}

\maketitle

\begin{abstract}
In this work, we study the problems of certifying and learning quantum Ising Hamiltonians. Our main contributions are as follows:

\textbf{Certification of Ising Hamiltonians.} We show that certifying an Ising Hamiltonian in normalized Frobenius norm via access to its time-evolution operator requires only $\widetilde O(1/\eps)$ time evolution. This matches the Heisenberg-scaling lower bound of $\Omega(1/\eps)$ up to logarithmic factors. To our knowledge, this is the first nearly-optimal algorithm for testing a Hamiltonian property. A key ingredient in our analysis is the Bonami Lemma from Fourier analysis.

\textbf{Learning Ising Gibbs states.} We design an algorithm for learning Ising Gibbs states in trace norm that is sample-efficient in all parameters. In contrast, previous approaches learned the underlying Hamiltonian (which implies learning the Gibbs state) but suffered from exponential sample complexity in the inverse temperature.

\textbf{Certification of Ising Gibbs states.} We give an algorithm for certifying Ising Gibbs states in trace norm that is both sample and time-efficient, thereby solving a question posed by Anshu (Harvard Data Science Review, 2022).

Finally, we extend our results on learning and certification of Gibbs states to general $k$-local Hamiltonians for any constant $k$.

% \fnote{Sketch of the abstract:    In this work, we consider the problems of testing and learning quantum Ising Hamiltonians. We prove: 
% \begin{enumerate}
%     \item When having access to the time evolution operator, one can certify Ising Hamiltonians with $\widetilde{O}(1/\eps)$ total evolution time. TBA, This is the first optimal property tester for quantum Hamiltonians. We prove it using the hypercontractive inequality of Fourier analysis. 
%     \item We show that an Ising Gibbs state can be learned in trace distance to accuracy $\eps$ with $O(\poly(n,\beta,\eps^{-1}))$ copies of the state. This exponentially improves the dependence on $\beta$ with respect to previous approaches, but it requires time exponential in $n.$
%     \item We show that an Ising Gibbs state can be time-efficiently certified in trace distance to accuracy $\eps$ with $O(n^2\beta^2/\eps^4)$ copies of the state. This answers a question by Anshu \cite{anshu2022some}.
% \end{enumerate}
% We also generalize 2. and 3. to Gibbs states of $k$-local Hamiltonians. }
\end{abstract}

\section{Introduction}
    With the rapid development of quantum hardware, the design of protocols to characterize its dynamics and its behavior at thermal equilibrium %\mnote{What do we mean by the thermal behavior of quantum hardware?} \fnote{ I mean their behaviour at theremal equilibrium. I have added equilibrium, is it better now?}
    has become increasingly more important \cite{bravyi2024high,liu2025robust}. Both aspects are ultimately governed by the system Hamiltonian, which has motivated an extensive literature on Hamiltonian learning \cite{Silva2011Practical,holzapfel2015scalable,Zubida2021Optimal,haah2022optimal,yu2023robust,Dutkiewicz.2023,huang2023heisenberg,li2023heisenberglimited,möbus2023dissipationenabled,franca2024efficient,Gu2022Practical,zhao2025learning,hu2025ansatz,anshu2021sample,rouze2023learning,onorati2023efficient,bakshi2023learning,ma2024learning,arunachalam2025testing,caro2023learning,castaneda2023hamiltonian} and, more recently, Hamiltonian testing \cite{aharonov2022quantum,laborde2022quantum,she2022unitary,bluhm2024hamiltonian,arunachalam2025testing,kallaugher2025hamiltonian,gao2025quantum,sinha2025improvedhamiltonianlearningsparsity}. %\mnote{Add \cite{} to the list of references?}

    A physically especially relevant class of quantum Hamiltonians is the family of Ising Hamiltonians, which can be written as a linear combination of Hamiltonians that act non-trivially on at most 2 particles \cite{ising1925beitrag}.\footnote{We abuse nomenclature here by identifying 2-local Hamiltonians with the subclass of quantum Ising Hamiltonians. %\mnote{Naive question: In what sense is this an abuse of nomenclature? Are Ising Hamiltonians usually only a subclass of 2-local Hamiltonians?}\fnote{People usually speak about the quantum Ising model as Hamiltonians with $2$-local terms involving only $Z$s and a $1$-local term involving only $X$s.}
    } 
    Both classical and quantum Ising models have been extensively studied (see for instance \cite{mccoy2012importance,daskalakis2019testing} for the classical case) since, despite their apparent simplicity, they are of fundamental importance in classical as well as quantum physics. For instance, they exhibit non-trivial quantum phase transitions \cite{dutta2010quantum,suzuki2012quantum}.
    Moreover, such Hamiltonians with 2-particle interactions have played a prominent role in quantum complexity theory \cite{oliveira2005complexity, kempe2006complexity, bravyi2017complexity}, with the corresponding $2$-local Hamiltonian problem proven to be QMA-complete; and they are known to be universal for quantum simulation \cite{cubitt2016complexity, cubitt2018universal}.
    % As Ising Hamiltonians are among the simplest models that exhibit non-trivial \mnote{Maybe we can specify in what way it is non-trivial? Are we thinking about phase transitions, frustration,...?}\fnote{Good point, can some of you fill this? I got this sentence from the paper of testing Ising models, I think. So if no one can add here, I can go back to that paper and add smt from there.} behavior, they have played an important role in quantum physics \cite{dutta2010quantum,suzuki2012quantum} and quantum complexity theory \cite{oliveira2005complexity,cubitt2016complexity,bravyi2017complexity}, and their classical counterparts have been extensively studied \cite{mccoy2012importance,daskalakis2019testing}.

    In this work, we propose algorithms to learn and to certify quantum Ising Hamiltonians, i.e., to test whether an unknown Ising Hamiltonian is equal to or far from a given target Ising Hamiltonian. In particular, when given access to the Hamiltonian through the time-evolution operator, we show that $\widetilde O(1/\eps)$ evolution time suffices for certification, yielding, to the best of our knowledge, the first optimal result in Hamiltonian property testing. For learning thermal states of Ising Hamiltonians, we give, to our knowledge, the first algorithm that is sample-efficient in all relevant parameters. Furthermore, for certifying such states, we give the first algorithm that is both sample- and time-efficient in all relevant parameters.
    %\mnote{Maybe also one sentence about the testing thermal states result.}\fnote{True! Done}

    \subsection{Results and Technical Overview}
    We will consider $n$-qubit Hamiltonians $H$, and for the rest of the introduction, we will assume that they are $2$-local. As such, their expansion in the Pauli basis is simply
    \begin{equation*}
        H=\sum_{P\in \{I,X,Y,Z\}^{\otimes n}:\, |P|\leq 2}h_P P\, ,
    \end{equation*}
    where $|P|$ is the number of sites where the Pauli string differs from identity. As two Hamiltonians that only differ in a multiple of the identity induce the same dynamics and thermal states, we assume without loss of generality that $h_{I^{\otimes n}}=\Tr[H]/2^n=0.$
    
    \subsubsection{Certification via access to the dynamics}
    If a quantum system is governed by a Hamiltonian $H$, then, according to the Schr\"odinger equation, its dynamics are determined by the unitary time evolution operator $U_H(t)=e^{-itH}$. By this, we mean that if the (mixed) state describing the system at time $0$ is $\rho,$ at time $t$ the state will have evolved to $U_H(t)\rho U_H^\dagger (t)$. Thus, a natural access model for Hamiltonians is to perform \emph{experiments} of the following kind: prepare a state $\rho,$ apply $U(t_1)$---that is, make a query to $U_H(t_1)$, which in a lab can be implemented by letting the system evolve for time $t_1$---, apply a unitary operator $V_1$ independent of $H$, query $U_H(t_2)$, apply a unitary operator $V_2$ independent of $H$, query $U_H(t_3)$, $\ldots$, and finally measure. There are several figures of merit to be optimized when performing a computational task in this access model. The one commonly considered the most important is the \emph{total evolution time}, which is the sum of all the times $t_i$ at which the algorithm queries $U_H(t_i).$ 

    As in prior work \cite{she2022unitary,bluhm2024hamiltonian,arunachalam2025testing,kallaugher2025hamiltonian,gao2025quantum}, we will assume that the Hamiltonians have bounded operator norm, $\norm{H}_{\op}\leq 1$, and we will consider the normalized Frobenius norm, given by $$\norm{H-H'}_{\bar F}=\sqrt{\Tr[(H-H')^2]/2^n},$$ as the distance in the space of Hamiltonians. 
    The normalized Frobenius norm is an average-case distance: if the normalized Frobenius norm between two Hamiltonians is small, then the expected values of observables measured on the two states generated by applying the time evolution of each Hamiltonian to a Haar-random state will be close \cite[Section 7.2]{ma2024learning}.
    
    Now, we are ready to state our first result (see \cref{theo:IsingCertificationTimeEvolution} for a formal and more detailed statement) on certifying Ising Hamiltonians from time evolution access.
    \begin{result}\label{result:IsingCertificationTimeEvolution}
        Let $H$ and $H_0$ be two $n$-qubit $2$-local Hamiltonians with $\norm{H}_{\op},\norm{H_0}_{\op}\leq 1$. Assume that $H_0$ is known in advance. Then, there is an algorithm with access to the time evolution of $H$ that only uses $\widetilde O(1/\eps)$ total evolution time, and, with high probability, determines whether $\norm{H-H_0}_{\bar F}\leq \eps$ or $\norm{H-H_0}_{\bar F}\geq 12\eps.$
    \end{result}

    \cref{result:IsingCertificationTimeEvolution} is optimal up to logarithmic factors, because $\Omega(1/\eps)$ evolution time is required to distinguish $H=\eps X$ from $H=-\eps X$ (see, for instance, \cite{kallaugher2025hamiltonian}). Several previous works considered testing Hamiltonian properties from time evolution access \cite{bluhm2024hamiltonian,arunachalam2025testing,kallaugher2025hamiltonian,gao2025quantum}, but the closest-to-optimal result for any Hamiltonian property testing task up to now was the $O(1/\eps^2)$ time evolution upper bound for testing locality given in \cite{kallaugher2025hamiltonian}, which is quadratically worse than the best lower bound $\Omega(1/\eps)$. 
    Thus, \cref{result:IsingCertificationTimeEvolution} is the first optimal (up to logarithmic factors) algorithm for testing a property of quantum Hamiltonians. Furthermore, we substantially improve upon the $\widetilde O(n^3/\eps)$-evolution-time algorithm for certifying Ising Hamiltonians that can be obtained as a special case of the $\widetilde O(s^{3/2}/\eps)$-evolution-time algorithm for certifying Hamiltonians supported on at most $s$ Pauli operators given in \cite[Theorem 5.5]{gao2025quantum}. In a concurrent work, a $O(1/\eps^2)$ evolution-time algorithm is given for the case where $H$ is an arbitrary Hamiltonian and $H_0=0$ \cite{sinha2025improvedhamiltonianlearningsparsity}. Compared to this, our result constitutes a quadratic improvement when $H$ is promised to be an Ising Hamiltonian.
    %\mnote{Should we also compare to the result you can get from \cite{sinha2025improvedhamiltonianlearningsparsity} here?}
    %\mnote{Just to double-check: Their result is for certifying more general Hamiltonians, but here we evaluate it for the 2-local case?}\fnote{Yes, true. I have rewritten it to make their result not look worse than it is.} 
    %Finally, we highlight that the only structural assumption necessary in \cref{result:IsingCertificationTimeEvolution} is 2-locality.
    %This is in contrast to previous works that assumed stronger requirements, such as bounded degree in the interaction graph of the Hamiltonian \cite{huang2023heisenberg,bakshi2024structure}. 
    %\mnote{Re the last sentence: The Hamiltonian testing papers don't make such assumptions. Are we trying to say here that the only previous works that achieve optimal results in some sense make such stronger assumptions?}\fnote{There is no previous testing work achieving optimal complexity. It is true that now that you say that sentence looks odd. Maybe, for the reasons that you point out, we can delete it. I guess that the comparison makes sense with respect to prior work on learning, but not on testing.}
    %\mnote{Agree. We are somehow comparing to prior work on optimal bounds for Hamiltonian learning, but I'm not sure whether we want this comparison here.}

    The proof of \cref{result:IsingCertificationTimeEvolution} relies on a novel application of the Bonami Lemma from Fourier analysis \cite{bonami1970etude,montanaro2008quantum}. We start by noting that we may assume query access to the time evolution operator of $\Delta H=H-H_0$ thanks to Trotterization, which allows us to approximate $e^{-it\Delta H}$ up to arbitrarily small error by making queries to $e^{-itH}$ and $e^{-itH_0}$.  We continue by considering the Taylor expansion of the time evolution operator and taking the trace, yielding
    \begin{equation*}
        \frac{\Tr[U_{\Delta H}(t)]}{2^n}=\underbrace{\frac{\Tr[I^{\otimes n}]}{2^n}}_{=1}+\underbrace{\frac{\Tr[-it\Delta H]}{2^n}}_{=0}-\frac{1}{2}\frac{\Tr{[(t\Delta H)^2]}}{2^n}+\sum_{l=3}^\infty \frac{1}{l!}\frac{\Tr{[(it\Delta H)^l]}}{2^n}\, .
    \end{equation*}
    In the above expression, we recognize two quantities: $\Tr[U_{\Delta H}(t)]/2^n$ is the Pauli coefficient $u_{I^{\otimes n}}$ of $U_{\Delta H}(t)$ corresponding to $I^{\otimes n}$, and $\Tr{[(t\Delta H)^2]}/2^n$ corresponds to $\norm{t\Delta H}_{\bar F}$. Hence, we arrive at 
    \begin{equation*}
        u_{I^{\otimes n}}=1-\frac{1}{2}\norm{t\Delta H}_{\bar F}^2+\underbrace{\sum_{l=3}^\infty \frac{1}{l!}\frac{\Tr{[(it\Delta H)^l]}}{2^n}}_{(*)}\, .
    \end{equation*}
    Assume for a moment that the error produced by the third term summand $(*)$ on the right-hand side was not there. In that case, we would have that if $t=1/(12\eps)$, then 
    \begin{align*}
        \norm{\Delta H}_{\bar F}\leq \eps\quad &\implies\quad\quad |u_{I^{\otimes n}}|\geq \frac{287}{288},\\
         \norm{\Delta H}_{\bar F}\geq 12\eps\quad &\implies\quad\quad |u_{I^{\otimes n}}|\leq \frac{144}{288}.
    \end{align*}
    In that case, it would suffice to estimate $|u_{I^{\otimes n}}|$ up to error $1/288$ to perform Hamiltonian certification. Such an estimation can be done with $O(1)$ queries to $U_{\Delta H}(t)$ (by performing Pauli sampling, or without memory using \cref{lem:memorylessPaulisampling} below), which would result in an algorithm for certification with $O(1)t=O(1/\eps)$ total evolution time, as desired. Thus, it remains to find a tool that allows to control the term $(*)$ even for long time scale $t=\Theta(1/\eps)$. That is exactly where the Bonami Lemma comes into play, allowing us to control higher-order moments with the second moment. More precisely, it states that
    \begin{equation*}
        \left(\frac{\Tr{[|\Delta H|^l]}}{2^n}\right)^{1/l}\leq l\left(\frac{\Tr{[(\Delta H)^2]}}{2^n}\right)^{1/2}=l\norm{\Delta H}_{\bar F}.
    \end{equation*}
    Using this, the term $(*)$ becomes negligible compared to $\norm{t\Delta H}_{\bar F}^2$, which permits us to reproduce the errorless approach above for the case of Ising Hamiltonians. 

    \subsubsection{Learning thermal states}
    There is plethora of results about learning quantum Hamiltonians from access to the associated Gibbs states \cite{anshu2021sample, haah2022optimal, rouze2023learning, onorati2023efficient, bakshi2023learning, Gu2022Practical,chen2025learning}, which, as noted in \cite[Remark 18]{anshu2021sample}, implies learning the Gibbs state itself. 
    However, this Hamiltonian learning-based approach to to the problem of Gibbs state learning inherits a $\Omega(e^{\beta})$ lower bound on the number of copies of the state \cite[Theorem 1.2]{haah2022optimal}.
    %the lower bound $\Omega(e^{\beta})$ copies of the state for Hamiltonian learning carries over through this approach to the problem of Gibbs state learning \cite[Theorem 1.2]{haah2022optimal}. 
    Here, we circumvent this caveat and obtain a learning algorithm for Gibbs states that is sample-efficient with respect to every parameter (see \cref{theo:GibbsStateLearning} for a formal statement). 
    \begin{result}\label{result:learningGibbsStates}
        Let $\rho_H(\beta)$ be the Gibbs state of an unknown $n$-qubit Ising Hamiltonian $H$ at temperature $\beta$ with $|h_P|\leq 1$ for every $P$. Then, there is an algorithm that, with probability at least $0.9,$ $\eps$-learns $\rho_H(\beta)$ in trace norm using only $\widetilde O(n^4\beta^2/\eps^4)$ single copies of the state. 
    \end{result}
    \cref{result:learningGibbsStates} can be generalized to $k$-local Hamiltonians, with $\widetilde O(n^{2k})$ sample-complexity instead of $O(n^2).$ Thus, in the case of $\beta=\poly(n)$ and $k=O(1)$, our algorithm achieves Gibbs state tomography with exponential speedup over general state tomography, which requires $\Theta(4^n)$ copies of the state \cite{o2016efficient,haah2017sample}. Notably, our result, in contrast to all the aforementioned prior works, only requires $k$-locality of the Hamiltonian, and no further assumptions (such as every qubit being acted on by a constant number of Pauli operators) are made. Sadly, our algorithm achieving \cref{result:learningGibbsStates} is not time-efficient, similarly to the first algorithm for learning quantum Hamiltonians from Gibbs states \cite{anshu2021sample}.

    The time-inefficiency is intrinsic to the $\eps$-covering net argument underlying the proof of \cref{result:learningGibbsStates}. The proof starts by establishing the following inequality, which is a consequence of Pinsker's inequality (see \cref{lem:dSKLGibbsStates} for a proof):
    \begin{equation}\label{eq:dsklintro}
        \norm{\rho_H(\beta)-\rho_{H'}(\beta)}_{\tr}
        \leq \sqrt{2 \beta \Tr[(\rho(\beta)-\rho'(\beta))(H'-H)]}
        =O(\beta n^2\max_{P:  |P|\le 2} |h_P-h_P'|)
    \end{equation}  
    %\textcolor{red}{AO : \cref{eq:trnormforGibbs0} goes better with the following explanation}
    for every pair of Ising Hamiltonians $H, H'$. This bound ensures that the set $$\mathcal S_\eta=\{\rho_H(\beta):\, H\in\mathcal H_\eta\}$$ of Gibbs states, where
    $$\mathcal H_\eta=\{H:H\text{ Ising Hamiltonian with }h_P\in \eta\mathbb Z\cap [-1,1]\ \forall P\},$$
    is an $\eps$-covering net of the set of Ising Gibbs states when taking $\eta$ of the order $\eps/(\beta n^2)$. Next, we note that the observables $H-H'$ for $H,H'\in\mathcal H_\eta$ are sums of $2$-local Pauli strings. Hence, via classical shadows \cite{huang2020predicting} (see \cref{theo:Shadows}), we can simultaneously obtain accurate estimates $\Delta_{H,H'}$ for all $\Tr[\rho(H-H')]$ in a sample-efficient manner, where $\rho$ is the state to be learned. If these estimates were exact and the state belonged to the net, then by \cref{eq:dsklintro} one would be able to identify the state. The rest of the proof consists of showing that, even if the state does not belong to the net and with error in the estimates, the state
    \begin{equation*}
        \rho'=\text{argmin}_{\rho'\in\mathcal S_\eta}\max_{H,H'\in\mathcal H_\eta}|\Delta_{H,H'}-\Tr[\rho'(H-H')]|
    \end{equation*}
    satisfies $\norm{\rho-\rho'}_{\tr}\leq \eps$ with high probability.
    
    \subsubsection{Certifying thermal states}
    We also show that quantum state certification of Ising Gibbs states can be made sample and time-efficient with respect to all parameters, resolving a question by Anshu \cite[Section 2]{anshu2022some} (see \cref{theo:certifyingGibbsStates} for a formal statement).

    \begin{result}\label{result:certifyingGibbsStates}
        Let $\rho_H(\beta)$ and $\rho_{H_0}(\beta)$ be the Gibbs states of an $n$-qubit Ising Hamiltonian $H$ and $H_0$ at temperature $\beta$ with $|h_P|,|(h_0)_P|\leq 1$ for every $P$. Then, there is an algorithm that, with probability at least $0.9,$ decides whether $\rho_{H}(\beta)=\rho_{H_0}(\beta)$ or $\norm{\rho_{H}(\beta)-\rho_{H_0}(\beta)}_{\tr}\geq \eps$ using only $\widetilde O(n^4\beta^2/\eps^4)$ single copies of the states. 
    \end{result}
    \cref{result:certifyingGibbsStates} can be generalized to $k$-local Hamiltonians, with $\widetilde O(n^{2k})$ sample-complexity instead of $O(n^2).$ Thus, in the case of $\beta=\poly(n)$ and $k=O(1)$, we have shown an exponential speedup for Gibbs state certification over general state certification, which requires $\Theta(2^n)$ copies of the state \cite{odonnell2015quantum, buadescu2019quantum}. Furthermore, the algorithm behind \cref{result:certifyingGibbsStates} is time-efficient in every parameter (in contrast with \cref{result:learningGibbsStates}). The proof of \cref{result:certifyingGibbsStates} is based on an inequality of the kind of \cref{eq:dsklintro}, which we expect to be useful in other scenarios, and which has seen applications in the classical literature~\cite{santhanam2012information,daskalakis2019testing}. 
    
    \subsection{Discussion and Outlook}

    In this work, motivated by the importance of quantum Ising Hamiltonians to various areas of quantum science, we have explored the tasks of certifying and learning these Hamiltonians. 
    First, we have given an algorithm for Ising Hamiltonian certification with optimal (up to logarithmic factors) total evolution time, thus providing the to our knowledge first optimal bound for any Hamiltonian property testing task in the time-evolution access model.
    Next, we have shifted our focus from the Hamiltonians themselves to the associated Gibbs states. For both learning and certification, this change of perspective allowed us to develop fully sample-efficient---and, in the case of certification, even time-efficient---algorithms.
    This in particular overcomes a known exponential-in-$\beta$ lower bound on learning Hamiltonians from access to copies of the Gibbs state, thus (re-)positioning Gibbs state learning and testing as tasks of independent interest alongside Hamiltonian learning and testing.
    
    % \mnote{Suggestion: Let's add one paragraph that recaps what we did and why it's important. I think we can emphasize:
    % \begin{enumerate}
    %     \item We have the first optimal result for any Hamiltonian property testing task.
    %     \item Conceptually, we propose that, in parallel to Hamiltonian learning and testing, also Gibbs state learning and testing are tasks of interest. For both learning and testing, we show that shifting the focus from the Hamiltonian to the Gibbs state allows for fully sample-efficient---and in the case of testing even time-efficient---algorithms.
    %     \begin{itemize}
    %         \item We already know that there is no algortihm for Hamiltonian learning from Gibbs states that scales polynomially in $\beta$, so here we have something rigorous to point out to demonstrate that our change of perspective is important. For Hamiltonian testing from Gibbs states, there is no such result in the literature, but we have the argument by Joe and Shivam. Do we want to include something along these lines to underscore that the change of perspective to Gibbs state testing is important? Or maybe including a result which we only have a half-baked proof for is too risky.
    %     \end{itemize}
    % \end{enumerate}
    % }
    
    We conclude this introduction by posing three open questions arising from our results:
    \begin{enumerate}
        \item Our nearly-optimal algorithm of \cref{result:IsingCertificationTimeEvolution} for certifying Ising models via access to the time evolution operator is the only one among our results that does not immediately generalize to $k$-local Hamiltonians for $k>2$. That is, our proof, which is based on the Bonami Lemma, breaks down for $k>2$ (see \cref{rem:wheretheproofbreaks}). Thus, it would be interesting to see whether this difference between $k=2$ and $k>2$ is fundamental or merely an artifact of our techniques. In particular, one may ask: Is it possible to certify $k$-local Hamiltonians with $\widetilde O(1/\eps)$ total evolution time for any constant $k$? For instance, can one employ tools developed to establish universality of two-qubit interactions for Hamiltonian simulation \cite{cubitt2018universal} to reduce the $k$-local to the $2$-local case?

        \item The seminal result of learning Hamiltonians via access to the Gibbs state of \cite{anshu2021sample} was only sample-efficient (with respect to $n$), and it was made time-efficient in a series of follow-up works \cite{haah2022optimal,bakshi2023learning}. Similarly, our \cref{result:learningGibbsStates} is, to our knowledge, the first algorithm for learning Gibbs states that is sample-efficient in all parameters. It is thus natural to wonder: Is there an algorithm for learning Gibbs states of Ising Hamiltonians that is both sample- and time-efficient in every parameter?

        \item \cref{result:certifyingGibbsStates} is already efficient in both sample and time complexity, but we lack a matching lower bound. Even for its classical counter-part \cite{daskalakis2019testing}, the precise complexity of Ising Gibbs states seems to be unknown.
        Thus, we ask: What is the optimal sample-complexity of certifying Ising Gibbs states?
    \end{enumerate}
    \paragraph{Acknowledgements.} A.B. was supported by the ANR project PraQPV, grant number ANR-24-CE47-3023. F.E.G. was supported by the European union’s Horizon 2020 research and innovation programme under the Marie Sk{\l}odowska-Curie grant agreement no. 945045, and by the NWO Gravitation project NETWORKS under grant no. 024.002.003. C.R.~is
supported by France 2030 under the French National Research Agency award number ``ANR-22-
PNCQ-0002''.

\section{Preliminaries}
We start by introducing some notation. $I,X,Y$ and $Z$ are the 1-qubit Pauli matrices, and a tensor product of these matrices is called a Pauli string.  Any matrix $A$ acting on $n$ qubits is a matrix of $(\C^{2\times 2})^{\otimes n}$. Such a matrix can be expressed as a linear combination of Pauli strings via its Pauli expansion $A=\sum_{P\in\{I,X,Y,Z\}^{\otimes n}}a_PP.$ Here, $a_P$ are the Pauli coefficients and they are determined by $$a_P=\frac{1}{2^n} \Tr[PA].$$ 
A Pauli string is called $k$-local if it acts as identity in all but at most $k$ qubits. The number of $k$-local Pauli strings is at most 
\begin{equation}\label{eq:numberofklocalPaulis}
    100n^k,
\end{equation}
because
\begin{align*}
        \sum_{l=0}^k 3^l\binom{n}{l}\leq \left\{\begin{array}{ll}
             (k+1)3^k\left(en/k\right)^k\leq 100 n^k             & \text{if } k<n/2 \\ 
            4^n\leq 20 n^{n/2}\leq 20n^k & \text{if }k\geq n/2
        \end{array}\right.,
\end{align*}
where we have used that $(3e/k)^k(k+1)<100$ and $4^n\leq 20 n^{n/2}$ for every $n,k\in\mathbb N$. Given a matrix $A$ acting on $n$ qubits, $\norm{A}_{\op}$ denotes the usual operator norm, i.e., the largest singular value of $A;$ $\norm{A}_{\tr}$ is the trace norm, i.e., the sum of the singular values of $A$; and $\norm{A}_{\bar F}=\Tr[A^\dagger A]/2^n$ is the normalized Frobenius norm. The Pauli strings are an orthonormal basis with respect to the inner product $\langle A,B\rangle=\Tr[A^{\dagger}B]/2^n.$ In particular, Parseval's identity states that
\begin{equation*}
    \norm{A}_{\bar F}=\sqrt{\sum_{P\in\{I,X,Y,Z\}^{\otimes n}}|a_P|^2}.
\end{equation*}
A more general version of Parseval's identity is Plancherel's identity, which states that 
\begin{equation*}
  \langle A,B\rangle\equiv   \frac{\Tr[A^\dagger B]}{2^n}=\sum_{P\in\{I,X,Y,Z\}^{\otimes n}}\bar{a}_Pb_P\, ,
\end{equation*}
where for $z\in\C$, $\bar z$ denotes the complex conjugate of $z.$ We use $\widetilde \Omega(\cdot)$ and $\widetilde O(\cdot)$ to hide polylogarithmic factors of the quantities inside the parentheses.

\subsection{Hamiltonians}
An $n$-qubit Hamiltonian is a self-adjoint matrix acting on $n$ qubits. In particular, a matrix $A$ is a Hamiltonian if and only if $a_P\in \mathbb R$ for every $P\in\{I,X,Y,Z\}^{\otimes n}$. A Hamiltonian $H$ is $k$-local if $h_P=0$ for every $P=P_1\otimes \dots \otimes P_n$ such that $|P|:=|\{i\in [n]:\ P_i\neq I\}|>k.$ Throughout this work, we will use the terms 2-local Hamiltonian and Ising Hamiltonian interchangeably. We will assume that Hamiltonians are traceless, meaning that $h_{I^{\otimes n}}=\Tr[H]/2^n$=0. This is without loss of generality, because two Hamiltonians that only differ in a multiple of identity determine the same time evolution operators and the same Gibbs states. 

\subsubsection{Access via time evolution operator} 
Hamiltonians govern the dynamics of (closed) quantum systems according to the Schr\"odinger equation. In particular, if a quantum system governed by a time-independent Hamiltonian $H$ and the state describing the system at time $0$ is $\rho,$ 
at time $t$ the state will have evolved to $U_H(t)\rho U_H^\dagger (t)$, where $U_H(t)=\exp(-itH)$ is the time evolution operator of $H$ at time $t$. 

Thus, a natural access model for Hamiltonians is to perform \emph{experiments} of the following kind: prepare a state $\rho,$ apply $U_H(t_1)$---that is, make a query to $U_H(t_1)$, which in a lab can be implemented by letting the system evolve for time $t_1$---, apply a unitary operator $V_1$ independent of $H$, query $U_H(t_2)$, apply a unitary operator $V_2$ independent of $H$, query $U_H(t_2)$,$\dots$ and finally measure. 
In this access model, there are different potentially relevant figures of merit. The one usually considered as the most important is the \emph{total evolution time}, which is the sum of all times $t_i$ at which the algorithm queries $U_H(t)$. Other figures of merit that we will also keep track of are the \emph{number of experiments}, the \emph{number of queries,} the \emph{time resolution} (i.e., the minimum time at which the algorithm queries the time evolution operator), the \emph{classical post-processing time}, and the number of \emph{ancilla qubits}.

Finally, our algorithms will also be robust to \emph{state-preparation and measurement (SPAM) error}. Following \cite[Definition 4]{ma2024learning}, an experiment suffers from an $\eps$-amount of SPAM error if the error channels applied after the initial state preparation and before the first query and the error channels after the last query and before the measurement induce in total $\eps$ error in diamond norm. We will say that an algorithm is \emph{robust} to an $\eps$ amount of SPAM error (or any other error) if the performance guarantees of the algorithm do not change in the presence of that error, maybe after increasing the complexities by constant factors. 

\subsection{Access via Gibbs state}
Hamiltonians also determine the equilibrium states of quantum systems. In particular, if a quantum system is governed by a Hamiltonian $H$, then the equilibrium state of the system at inverse temperature $\beta>0$ is  the \emph{Gibbs state} given by $\rho(\beta)=e^{-\beta H}/\Tr[e^{-\beta H}]$. 

An alternative access model for Hamiltonians is hence to perform measurements on copies of the Gibbs state of the Hamiltonian. The main figure of merit in this model is the \emph{sample complexity,} i.e., the number of copies of the Gibbs state that the algorithm accesses. Other important figures of merit that we will keep track of are the \emph{maximum number of copies that the algorithm measures coherently} and the \emph{classical post-processing time}. In particular, we say that an algorithm uses \emph{single copies} of the state if it measures one copy of the state at a time.

All of our results in this access model use the following upper bounds on the trace distance between Gibbs states, which are well-known in the classical literature \cite{santhanam2012information,daskalakis2019testing}, and similar bounds have been used in the quantum literature \cite{anshu2021sample,fanizza2024efficient}. 

\begin{lemma}\label{lem:dSKLGibbsStates}
    Let $\rho(\beta)$ and $\rho'(\beta)$ be Gibbs states of two $k$-local Hamiltonians $H$ and $H'$ acting on $n$ qubits. Then,
    \begin{equation}\label{eq:trnormforGibbs0}
        \norm{\rho(\beta)-\rho'(\beta)}_{\tr}\leq \sqrt{2 \beta \Tr[(\rho(\beta)-\rho'(\beta))(H'-H)]}.
    \end{equation}
    In particular,
    \begin{equation}\label{eq:trnormforGibbs1}
        \norm{\rho(\beta)-\rho'(\beta)}_{\tr}\le  200\beta n^k \sup_{|P|\leq k}|h_P-h_P'|.
    \end{equation}
    Furthermore, if $|h_P|,|h'_P|\leq 1$ for every $P\in\{I,X,Y,Z\}^{\otimes n}$, then 
    \begin{equation}\label{eq:trnormforGibbs2}
        \norm{\rho(\beta)-\rho'(\beta)}_{\tr}\le \sqrt{400\beta n^k \sup_{|P|\leq k}2^n|\rho(\beta)_P-\rho'(\beta)_P|}.
    \end{equation}
\end{lemma}
\begin{proof}
    We start using Pinkser inequality to upper bound the trace norm as 
    \begin{align*}
        \norm{\rho(\beta)-\rho'(\beta)}_{\tr}\le\sqrt{2\Tr[\rho(\beta)(\log\rho(\beta)-\log\rho'(\beta))]+2\Tr[\rho'(\beta)(\log\rho'(\beta)-\log\rho(\beta))]}.
    \end{align*}
    Now, expanding the right-hand side and using that $\log \rho(\beta)=-\beta H-Z(\beta)$, where $Z(\beta)=\Tr[e^{-\beta H}]$, we arrive at 
    %\mnote{I agree with Andreas's comment, I think it should either be $-\beta$ or maybe $H'-H$. Of course, I might have made a sign error...}\fnote{True! Fixed.}
    \begin{align}
        \norm{\rho(\beta)-\rho'(\beta)}_{\tr}&\le \sqrt{2\beta\Tr[(\rho(\beta)-\rho'(\beta))(H'-H)]}.\label{eq:trnormforGibbs3}
    \end{align}
    This proves \cref{eq:trnormforGibbs0}.

    Now, we focus on proving \cref{eq:trnormforGibbs1}. On the one hand, using \cref{eq:trnormforGibbs3} and that $|\Tr[A^\dagger B]|\leq \norm{A}_{\tr} \norm{B}_{\op}$ we get 
    \begin{align*}
        \norm{\rho(\beta)-\rho'(\beta)}_{\tr}
        &\le \sqrt{2\beta\norm{\rho(\beta)-\rho'(\beta)}_{\tr}\norm{H-H'}_{\op}}\, ,%\leq \sqrt{2\beta\norm{\rho(\beta)-\rho'(\beta)}_{\tr}\norm{H-H'}_{\op}},
    \end{align*}
    so 
    \begin{align}
        \norm{\rho(\beta)-\rho'(\beta)}_{\tr}
        &\le 2\beta\norm{H-H'}_{\op},\label{eq:1}
    \end{align}
    On the other hand, by triangle inequality and \cref{eq:numberofklocalPaulis}, we have that %\mnote{Similarly to Andreas, I think that the next line uses an implicit assumption, see Overleaf comment.} \fnote{I see now! Fixed!}
    \begin{align*}
        \norm{H-H'}_{\op}&\leq  100 n^k \sup_{|P|\leq k}|h_P-h_P'|,
    \end{align*}
%    \begin{align*}
%        \norm{H-H'}_{\op}&\leq \sum_{l=0}^k 3^l{n\choose l}\sup_{|P|\leq k}|h_P-h_P'|\leq (k+1)3^k{n \choose k}\sup_{|P|\leq k}|h_P-h_P'|\\
%        &\leq (k+1)3^k\left(\frac{en}{k}\right)^k \sup_{|P|\leq k}|h_P-h_P'|\leq  100 n^k \sup_{|P|\leq k}|h_P-h_P'|,
%    \end{align*}
    which combined with \cref{eq:1} proves \cref{eq:trnormforGibbs1}. 

    Now, we focus on proving \cref{eq:trnormforGibbs2}. Using \cref{eq:trnormforGibbs3} and Plancherel's identity we arrive at 
    \begin{align*}
        \norm{\rho(\beta)-\rho'(\beta)}_{\tr}&\le\sqrt{2\beta\sum_{|P|\leq k}2^n(\rho(\beta)_P-\rho'(\beta)_P)(h_P-h'_P)}.
    \end{align*}
    Hence, as $|h_P|,|h_P'|\leq 1$ and by \cref{eq:numberofklocalPaulis}, we have that 
    \begin{align*}
        \norm{\rho(\beta)-\rho'(\beta)}_{\tr}&=\sqrt{400\beta n^k\sup_{|P|\leq k}2^n|\rho(\beta)_P-\rho'(\beta)_P|},
    \end{align*}
    which proves \cref{eq:trnormforGibbs2}.
\end{proof}

\subsubsection{Trotterization}
Given access to $e^{-itA}$ and $e^{-itB}$ for two Hamiltonians $A$ and $B$ and arbitrary times $t$, Trotterization allows us to implement $e^{-it(A+B)}$ up to arbitrary error while also preserving the total time evolution and without using extra qubits. Thus, to analyze the number of experiments and the total time evolution required by our algorithms, if we have access to $e^{-itA}$ and $e^{-itB}$, we may assume access to $e^{-it(A+B)}$. However, the number of queries and the time resolution change. To be more precise, we will use the following result. 

\begin{theorem}\cite[Corollary 2]{childs2021theory}\label{theo:trotterization}
    Let $t>0$, let $\eps>0$, let $H,H_0$ be Hamiltonians acting on $n$-qubits, and let $c=\max\{\norm{H}_{\op},\norm{H_0}_{\op}\}$. Let $l=\left\lceil O\left(\sqrt{(ct)^3/\eps_{\operatorname{Trott}}}\right)\right\rceil$ and define $V=(e^{-itH/2l}  e^{itH_0/l}e^{-itH/2l})^l$. Then, 
    % \begin{equation*}
    %     \norm{e^{-it(H-H_0)}\cdot e^{it(H-H_0)}-V\cdot V^{\dagger}}_{\diamond}\leq \eps.
    % \end{equation*}
        \begin{equation*}
        %\textcolor{red}{
        \norm{e^{-it(H-H_0)}-V}_{\op}
        %}
        \leq \eps_{\operatorname{Trott}}.
    \end{equation*}
\end{theorem}
\subsubsection{Useful subroutines}

We will use the following lemma that was proved in \cite[Lemma 3.3]{arunachalam2024testing}.\footnote{We note that in \cite[Lemma 3.3]{arunachalam2024testing} the authors only explicitly analyze the query complexity of their algorithm, but the analysis of the remaining figures of merit is straightforward.} Before stating it, we recall that a stabilizer subgroup of the group of Pauli matrices $\mathcal{S}\subseteq\{I,X,Y,Z\}^{\otimes n}$ is an abelian subgroup that does not contain $-I$. A stabilizer state corresponding to a stabilizer subgroup $\mathcal{S}$ of dimension $k\le n$ is defined as
\begin{align*}
\rho_{\mathcal{S}}:=\frac{1}{2^n}\sum_{P\in\mathcal{S}}P\,.
\end{align*}

\begin{lemma}\label{lem:memorylessPaulisampling}
    Let $U$ be an $n$-qubit unitary, and let $\eps,\delta>0$. There is a memory-less algorithm that makes $O\big(\log(1/\delta)/\eps^2\big)$ experiments that provides an estimate $|u'_{I^{\otimes n}}|^2$ such that %\mnote{Sorry, naive question: Do we really only care about $u_{I^{\otimes n}}$ and not also about other Pauli coefficients?}\fnote{Yes, we care only about this one Pauli coefficient}
    $$
    \big||u_{I^{\otimes n}}|^2-|u_{I^{\otimes n}}'|^2\big|\leq\eps
    $$
    with probability $\geq 1-\delta$. Furthermore, the algorithm makes only one query to $U$ per experiment, only stabilizer states, and only performs Clifford measurements. In addition, it is robust to $\eps/3$ amount of SPAM errors and $\eps/3$ error in diamond norm per query of $U$%\textcolor{red}{I find this a bit confusing, also where is this done?}
    , and requires only $O\big(\log(1/\delta)/\eps^2\big)$ classical post-processing time.
\end{lemma}

We will also need to perform classical shadow tomography.

\begin{theorem}[Clifford shadows \cite{huang2020predicting}]\label{theo:Shadows}
    Let $\rho$ be an $n$-qubit state and let $k\in\mathbb N$, $\eps>0$ and $\delta>0.$ Then, performing random Pauli measurements on $$O
    \left(\frac{3^kk\log(n/\delta)}{\eps^2}\right)$$ single copies of $\rho$ suffices to obtain estimates $\widetilde \rho_P$ that with probability $\geq 1-\delta$ satisfying $$2^n|\rho_P-\widetilde \rho_P|\leq \eps$$ for every $|P|\leq k$. The classical post-procesing time is $O
    \left((3n)^kk\log(n/\delta)/\eps^2\right)$.
\end{theorem}
\subsection{Bonami Lemma}
We will use the quantum version of Bonami Lemma~\cite{bonami1970etude} proved by Montanaro and Osborne~\cite[Corollary 8.9]{montanaro2008quantum}. 
\begin{theorem}\label{lem:BonamiLemma}
    Given a $k$-local Hamiltonian $H$ on $n$-qubits and $l\geq 2,$ it holds that $$\left(\frac{\Tr[|H|^l]}{2^n}\right)^{1/l}\leq l^{k/2}\left(\frac{\Tr[ H^2]}{2^n}\right)^{1/2}.$$ 
\end{theorem}

\section{Hamiltonian certification via access to time-evolution}
In this section, we propose an algorithm that uses access to time-evolution to certify whether an unknown Ising Hamiltonian $H$ is close to or far from a known Ising Hamiltonian $H_0$. We prove that $\widetilde O(1/\eps)$ total evolution time suffices to solve this problem optimally (see \cref{theo:IsingCertificationTimeEvolution}). %\textcolor{red}{optimality}

We start by proving that such a certification is possible in a more restricted setting where both Hamiltonians are promised to be not too far from each other (see \cref{lem:IsingCertificationTimeEvolutionBasisCase}). The result in the general setting proceeds by iterating this restricted case.

\begin{algorithm}[H]
\caption{Hamiltonian certification subroutine}\label{alg:quantum_data1}

\begin{algorithmic}[1]
\Require Parameters $\delta\in(0,1)$, $\eps>0$, time evolution access to $H$ and $H_0$ 
\State Implement unitary $V$ from \Cref{theo:trotterization}, with $\eps_{\operatorname{Trott}}=\tfrac{1}{19200 e^6 C^2}$ and $t=1/(60\eps e^3 C)$, where $C=\sum_{l=0}^\infty (e^{-2})^l$.
\State Use the algorithm of \cref{lem:memorylessPaulisampling} to obtain $|v_{I^{\otimes n}}'|^2$, that, with probability $\geq 1-\delta$, is an $\tfrac{1}{4800 e^6 C^2}$-estimate of $|v_{I^{\otimes n}}|^2$ 
\If{$|v_{I^{\otimes n}}'|^2 \leq 1 - \tfrac{23}{2400 e^6 C^2}$}
  \State \Return ``FAR''
\Else
  \State \Return ``CLOSE''
\EndIf
\end{algorithmic}
\end{algorithm}

\begin{lemma}\label{lem:IsingCertificationTimeEvolutionBasisCase}
    Let $H$ and $H_0$ be $n$-qubit Ising Hamiltonians, where $H_0$ is known and $H$ can be accessed via its time evolution operator, and denote $\Delta H:=H-H_0$. Let $\eps,\delta>0.$ Let $C_{\op}\geq 1$ be such that $\norm{H_0}_{\op},\norm{H}_{\op}\leq C_{\op}$. Assume that $\norm{\Delta H}_{\bar F}\leq 15\eps$. Then, \Cref{alg:quantum_data1} uses $O(\log(1/\delta)/\eps)$ total evolution time to test whether $\norm{\Delta H}_{\bar F}< \eps$ or $\norm{\Delta H}_{\bar F}> 12\eps.$ %\mnote{Because the color in the proof might be a bit hidden: $\leq$ and $\geq$ or $<$ and $>$?} \fnote{Good catch, fixed}
    
        Moreover, even if none of the two promises is satisfied, with probability $1-\delta$, we have that if the algorithm outputs ``FAR'', then $\norm{\Delta H}_{\bar F}\geq\eps$, and if it outputs ``CLOSE'', then $\norm{\Delta H}_{\bar F}\leq 12\eps.$ 

    Furthermore, the algorithm uses no ancilla qubits, it makes $O(\log(1/\delta))$ experiments, it makes $O((C_{\op}/\eps)^{3/2}\cdot \log(1/\delta))$ queries, the time resolution is $\Omega(\eps^{1/2}/C_{\op}^{3/2})$, the algorithm is robust to a constant amount of SPAM errors, and the classical post-processing time is $O(\log(1/\delta))$.
\end{lemma}
\begin{proof}

First, by the Trotterization of \cref{theo:trotterization}, for $U=e^{-it\Delta H}$,
\begin{align*}
|v_{I^{\otimes n}}-u_{I^{\otimes n}}|=\frac{1}{2^n}\Big|\Tr(I^{\otimes n}[U-V])\Big|\le \|U-V\|_{\op}\le \eps_{\operatorname{Trott}}\, ,
\end{align*}
so the estimate $|v'_{I^\otimes n}|^2$ of $|v_{I^{\otimes n}}|^2$, in the presence of $\eps_{\mathrm{SPAM}}$ error of at most $\tfrac{1}{9600 e^6 C^2}$,  is a $(\tfrac{1}{4800 e^6 C^2}+2\eps_{\operatorname{Trott}}+\eps_{\text{SPAM}}=\tfrac{1}{2400 e^6 C^2})$-estimate of $u_{I^{\otimes n}}$. From this estimate, we show correctness and then perform a complexity analysis.

       \textbf{Correctness analysis.} We aim to prove that with probability $\geq 1-\delta,$ $\norm{\Delta H}_{\bar F}> 12\eps\implies $ we output ``FAR'', and  $\norm{\Delta H}_{\bar F}< \eps\implies $ we output ``CLOSE''. We start by noting that by Taylor expansion 
       \begin{align*}
            \left|u_{I^{\otimes n}}-\left(1-\frac{1}{2}\frac{t^2\Tr[\Delta H^2]}{2^n}\right)\right|\leq \sum_{l=3}^\infty \frac{t^l}{l!}\frac{\Tr[|\Delta H|^l]}{2^n}.
        \end{align*}
        Note that we can identify $\norm{\Delta H}_{\bar F}$ in the above expression, so we can rewrite
        \begin{align*}
            \left|u_{I^{\otimes n}}-\left(1-\frac{(t\norm{\Delta H}_{\bar F})^2}{2}\right)\right|\leq \sum_{l=3}^\infty \frac{t^l}{l!}\frac{\Tr[|\Delta H|^l]}{2^n}.
        \end{align*}
        Now, we can upper-bound the right-hand side as
        \begin{align}\label{eq:wheretheproofbreaks}
            \sum_{l=3}^\infty \frac{t^l}{l!}\frac{\Tr[|\Delta H|^l]}{2^n}&\leq \sum_{l=3}^\infty \frac{t^l}{l!}\left(l\left(\frac{\Tr[\Delta H^2]}{2^n}\right)^{1/2}\right)^{l}=\sum_{l=3}^\infty (t\norm{\Delta H}_{\bar F})^l\frac{l^{l}}{l!}\\
            &\leq \sum_{l=3}^\infty (t\norm{\Delta H}_{\bar F})^le^{l}=  e^3(t\norm{\Delta H}_{\bar F})^3\sum_{l=0}^\infty (et\norm{\Delta H}_{\bar F})^l\nonumber \\
    &\leq  e^3(t\norm{\Delta H}_{\bar F})^3\underbrace{\sum_{l=0}^\infty (e^{-2})^l}_{=C},\nonumber
    \end{align}
    where in the first line we have used \cref{lem:BonamiLemma}, and in the second line that $l^l\leq e^l l!,$ and in the third line that, by assumption $\|\Delta H\|_{\bar F}\le 15\eps$, so that $et\norm{\Delta H}_{\bar F}\leq e^{-2}$. Thus, we have shown that 
    \begin{align*}
        \left|u_{I^{\otimes n}}-\left(1-\frac{(t\norm{\Delta H}_{\bar F})^2}{2}\right)\right|\leq Ce^3(t\norm{\Delta H}_{\bar F})^3.
    \end{align*}
    Now, as $Ce^3t\norm{\Delta H}_{\bar F}\leq 1/4$, we have that 
    \begin{align*}
        1-\frac{3}{4}(t\norm{\Delta H}_{\bar F})^2\leq |u_{I^{\otimes n}}|\leq 1-\frac{1}{4}(t\norm{\Delta H}_{\bar F})^2.
    \end{align*}
    Taking squares we arrive at 
    \begin{align*}
        1-\frac{3}{2}(t\norm{\Delta H}_{\bar F})^2+\frac{9}{16}(t\norm{\Delta H}_{\bar F})^4\leq |u_{I^{\otimes n}}|^2\leq 1-\frac{1}{2}(t\norm{\Delta H}_{\bar F})^2+\frac{1}{16}(t\norm{\Delta H}_{\bar F})^4.
    \end{align*}
    Hence, recalling that $t\norm{\Delta H}_{\bar F}\leq 1$ we conclude that 
    \begin{align*}
        1-\frac{3}{2}(t\norm{\Delta H}_{\bar F})^2\leq |u_{I^{\otimes n}}|^2\leq 1-\frac{1}{4}(t\norm{\Delta H}_{\bar F})^2.
    \end{align*}
    Hence, we have that 
    \begin{align*}
        \norm{\Delta H}_{\bar F}< \eps\quad &\implies \quad |u_{I^{\otimes n}}|^2\geq 1-\frac{3}{2}\frac{1}{(60e^3C)^2}=1-\frac{1}{2400e^6C^2},\\
        \norm{\Delta H}_{\bar F}> 12\eps\quad &\implies \quad |u_{I^{\otimes n}}|^2\leq 1-\frac{1}{4}\frac{12^2}{(60e^3C)^2}=1-\frac{24}{2400e^6C^2}.
    \end{align*}
    Thus, since $|v'_{I^{\otimes n}}|^2$ is a $(1/2400e^6C^2)$-estimate of $|u_{I^{\otimes n}}|^2$, then 
    \begin{align*}
        \norm{\Delta H}_{\bar F}< \eps\quad &\implies |v'_{I^{\otimes n}}|^2\geq 1-\frac{2}{2400e^6C^2}\quad \implies \quad \text{we output ``CLOSE'',}
        \\
        \norm{\Delta H}_{\bar F}> 12\eps\quad &\implies \quad |v_{I^{\otimes n}}'|^2\leq 1-\frac{23}{2400e^6C^2}\quad \implies \quad \text{we output ``FAR'',}
    \end{align*}
    as desired. 
    
    \textbf{Complexity analysis.} By \cref{lem:memorylessPaulisampling}, we need to make  $O\big(\log(1/\delta)\big)$ queries to $V$, where each query corresponds to $l=O\big((C_{\op}/\eps)^{3/2}\big)$ queries to the Hamiltonian evolution $H$ at time resolution $\Omega\Big(\eps^{1/2}/C_{\operatorname{op}}^{3/2}\Big)$ by virtue of \Cref{theo:trotterization}. Hence, the total number of queries to $H$ is $O((C_{\op}/\eps)^{3/2}\cdot \log(1/\delta))$, the total time evolution required then scales like $O(\eps^{-1}\log(1/\delta))$. %\textcolor{red}{classical post-processing correct? convince myself of SPAM robustness}

    % \textcolor{red}{=========}
    
    %  For the total evolution time, by virtue of \cref{theo:trotterization}, we may assume that we have access to the time evolution of $\Delta H$. Thanks to \cref{lem:memorylessPaulisampling}, we just need to make $O(\log(1/\delta))$ queries to $e^{-it\Delta H}$ for $t=\Theta(1/\eps)$, so the total time evolution is $O(1/\eps).$ We also note that as the algorithm of \cref{lem:memorylessPaulisampling} in our case is robust to an amount of $(1/600e^6C^2)/3=\Theta(1)$ SPAM error. It is also robust to that amount of error in the implementation of $U$, and by \cref{theo:trotterization} we can implement $U$ with that amount of error with $O((\max\{\norm{H}_{\op},\norm{H_0}_{\op}\}/\eps)^{3/2})$ queries to time evolution operator of $H$ at time $\Theta(\eps^{1/2}/(\max\{\norm{H}_{\op},\norm{H_0}_{\op}\})^{3/2})$. Then, the claimed number of queries and the time resolution follow. The classical post-processing time and the number of ancilla qubits follow from \cref{lem:memorylessPaulisampling}.
\end{proof}

Our first main result concerning the optimal certification of quantum Ising Hamiltonians follows by iterating \Cref{alg:quantum_data1}.

\begin{algorithm}[H]
\caption{Hamiltonian certification via time-evolution}\label{alg:quantum_data2}

\begin{algorithmic}[1]
\Require Time evolution access to $H_0$ and $H$ with $\norm{H_0}_{\bar F},\norm{H}_{\bar F}\leq C_{\bar F}$ and $\norm{H_0}_{\op },\norm{H}_{\op}\leq C_{\op}$,  parameters $\delta\in (0,1)$ and $\eps\in (0,C_{\bar F})$% with $\|H_0\|_{\bar F}, \|H\|_{\bar F}\le C_{\bar F}$%, unitary $V$ from \Cref{theo:trotterization}, with $\eps_{\operatorname{Trott}}=\tfrac{1}{9600 e^6 C^2}$ and $t=1/(60\eps e^3 C)$, $C=\sum_{l=0}^\infty (e^{-2})^l$.
\State Set $l=L$, $L=\lceil \log_{15/12}(2C_{\bar F}/15\eps) \rceil$
\State Use \Cref{alg:quantum_data1} with $\eps_l=(15/12)^l\eps$ and $\delta_l=\delta/(L+1)$. 
\If{``FAR''}
  \State \Return ``FAR''
\ElsIf{``CLOSE'' and $l>0$}
  \State Set $l=l-1$ and go back to Step 2.
 \ElsIf{$l=0$}
   \State Terminate and output ``CLOSE''.
\EndIf
\end{algorithmic}
\end{algorithm}

   % \textbf{The algorithm.} We will use \cref{lem:IsingCertificationTimeEvolutionBasisCase} iteratively at most $L=\lceil \log_{15/12}(2C_{\bar F}/15\eps) \rceil$. Let $l=L$. We will iterate the following procedure 
   % \begin{enumerate}
   %     \item Run the algorithm of \cref{lem:IsingCertificationTimeEvolutionBasisCase} with $\eps_l=(15/12)^l\eps$ and $\delta_l=\delta/(L+1)$. 
   %     \begin{enumerate}
   %         \item If it outputs far, then we terminate and output far,
   %         \item If it outputs close and $l>0$ continue and set $l$ to $l-1.$ If $l=0$, we terminate and output close.
   %     \end{enumerate}
   % \end{enumerate}

\begin{theorem}\label{theo:IsingCertificationTimeEvolution}
    Let $H$ and $H_0$ be $n$-qubit Ising Hamiltonians, where $H_0$ is known and $H$ can be accessed via its time evolution operator. Let $C_{\op}\geq 1$ be such that $\norm{H_0}_{\op},\norm{H}_{\op}\leq C_{\op}$, and let $C_{\bar F}\geq 1$ be such that $\norm{H_0}_{\bar F},\norm{H}_{\bar F}\leq C_{\bar F}$. Let $\delta>0$ and $\eps\in (0,C_{\bar F})$. Then,  \Cref{alg:quantum_data2} uses $\widetilde O(\log(C_{\bar F}/\delta)/\eps)$ total evolution time to test whether $\norm{\Delta H}_{\bar F}\leq \eps$ or $\norm{\Delta H}_{\bar F}\geq 12\eps,$ promised that one of the two is satisfied.

    Furthermore, the algorithm uses no ancilla qubits, it makes $\widetilde O(\log(C_{\bar F}/\eps)\cdot \log(1/\delta))$ experiments, it makes $\widetilde O((C_{\op}/\eps)^{3/2}\cdot \log(C_{\bar F})\cdot \log(1/\delta))$ queries, the time resolution is $\widetilde{\Omega}(\eps^{1/2}/C_{\op}^{3/2})$, the algorithm is robust to a constant amount of SPAM errors, the classical post-processing time is $\widetilde O(\log(C_{\bar F}/\eps)\log(1/\delta))$. %Above $\widetilde{\Omega}$ and $\widetilde{O}$ hides logarithmic dependencies of each of the parameters involved. \fnote{I have added a definition of this notation in the preliminaries}
\end{theorem}

\begin{proof}
  As above, we first show correctness, then perform a complexity analysis.

   \textbf{Correctness analysis.} In the iteration with $l=L$ we have that $\eps_l\ge 2C_{\bar F}/15,$ so $\norm{\Delta H}_{\bar F}\leq 2C_{\bar F}\le 15\eps_l$. Thus, by \cref{lem:IsingCertificationTimeEvolutionBasisCase} with probability $\geq 1-\delta/(L+1)$ we have the following: On the one hand, if the output of \Cref{alg:quantum_data1} on that iteration is ``FAR'', then $\norm{\Delta H}_{\bar F}\geq \eps_l=(15/12)^l\eps\geq (15/12)\eps$, so we are correct if we terminate and output ``FAR''. On the other hand, if the output of \Cref{alg:quantum_data1} on that iteration is ``CLOSE'', then $\norm{\Delta H}_{\bar F}\leq 12 \eps_l=12\cdot (15/12)^l\eps\leq 15\eps_{l-1}$, so we are in conditions of applying \Cref{alg:quantum_data1} with the parameters $\eps_{l-1},\delta_{l-1}.$ We can iterate this argument up to the iteration with $l=0$. If we arrive at the iteration of $l=0$, then we know that $\norm{\Delta H}_{\bar F}\leq 15\eps$, so this iteration of \Cref{alg:quantum_data1} will output the correct answer. Finally, we note that as every iteration succeeds with $1-\delta/(L+1)$ and there is at most $L+1$ iterations, we have that the algorithm succeeds with probability $\geq 1-\delta.$

   \textbf{Complexity analysis.} The complexity analysis follows from the fact that we just have to run \Cref{alg:quantum_data1} for $L=O(\log(C_{\bar F}/\eps))$ times with parameters $\eps'=\Omega(\eps)$ and $\delta'=\delta/(L+1)=\Omega(\delta/\log(C_{\bar F}/\eps))$.

  %\textcolor{red}{check SPAM robustness and classical postprocessing} 
\end{proof}
  \begin{remark}\label{rem:wheretheproofbreaks}
      Our proof technique breaks down when considering $k$-local Hamiltonians for $k>2.$ In that case, instead of \cref{eq:wheretheproofbreaks} we would have 
      \begin{align}\label{eq:proofbreaking}
            \sum_{l=3}^\infty \frac{t^l}{l!}\frac{\Tr[|\Delta H|^l]}{2^n}&\leq \sum_{l=3}^\infty \frac{t^l}{l!}\left(l^{k/2}\left(\frac{\Tr[\Delta H^2]}{2^n}\right)^{1/2}\right)^{l}=\sum_{l=3}^\infty (t\norm{\Delta H}_{\bar F})^l\frac{l^{lk/2}}{l!}.
    \end{align}
    %\textcolor{red}{ao : do we have the reverse inequality, this will make the argument clearer.}
    However, for $k>2$ we have that $l^{lk/2}/l!=\Omega(l^{l/2})$, so the series on the right-hand side is lower bounded as 
    \begin{align*}
        \sum_{l=3}^\infty (t\norm{\Delta H}_{\bar F})^l\frac{l^{lk/2}}{l!}\geq \sum_{l=3}^\infty (t\norm{\Delta H}_{\bar F}l^{1/2})^l,
    \end{align*}
    which diverges. Thus, \cref{eq:proofbreaking} becomes $$ \sum_{l=3}^\infty \frac{t^l}{l!}\frac{\Tr[|\Delta H|^l]}{2^n}\leq \infty, $$ which is meaningless.
  \end{remark}

\section{Learning and certifying Gibbs states}
\subsection{Learning Gibbs states}
In this section we propose a fully-sample-efficient protocol to learn Gibbs states, i.e., an algorithm whose sample-complexity is at most polynomial in all relevant parameters. First, we show that the following set is an $\eps$-covering net for the set of Gibbs states coming from a $k$-local Hamiltonian with bounded Pauli coefficients:
\begin{equation}
    \mathcal S_{\eps,k,n,\beta}=\left\{ e^{-\beta H}/\Tr[e^{-\beta H}]\ :\ H\in \mathcal H_{\eps,k,n,\beta} \right\},
\end{equation}
where 
\begin{equation*}
    \mathcal H_{\eps,k,n,\beta}=\left\{H: H=\sum_{|P|\leq k} h_{P}P, \; h_{P}\in \eta\mathbb{Z}\cap [-1,1]\right\}
\end{equation*}
and $\eta=\eta_{\eps,k,n,\beta} = \eps/(200\beta n^k)$. %\mnote{As we later take $\max\{\beta, 1\}$, should we also have that maximum in the definition of $\eta$?}
\begin{lemma}\label{lem:net}
    Let $H$ be a $k$-local Hamiltonian acting on $n$ qubits with $|h_P|\leq 1$ for every $P\in\{I,X,Y,Z\}^{\otimes n}$. Then, there exists $\rho\in \mathcal S_{\eps,k,n,\beta}$ such that $\norm{\rho(\beta)-\rho}_{\tr}\leq \eps.$
\end{lemma}
\begin{proof}
    Given $P\in\{I,X,Y,Z\}^{\otimes n}$, let $h_P'$ be the element of $\eta\mathbb Z\cap[-1,1]$ that is closest to $h_P$. Let $H'=\sum h_P'P.$ Then, $\rho=e^{-\beta H'}/\Tr[e^{-\beta H'}]$ belongs to $\mathcal S_{\eps,k,n,\beta}.$ Also, by \cref{lem:dSKLGibbsStates} we have that 
    \begin{align*}
    \norm{\rho(\beta)-\rho}_{\tr}=200\beta n^k\max_{|P|\leq k}|h_P-h_P'|\leq 200\beta n^k\eta = \eps\, ,
    \end{align*}
    where we used the choice of $\eta$ in the last step.
\end{proof}

\noindent Next, we introduce some observables whose expected value will allow us to determine which element of the net is closest to the unknown state. Note that $$|\mathcal H_{\eps,k,n,\beta}|=|\mathcal S_{\eps,k,n,\beta}|=(2/\eta)^{O(n^k)}=(n^{k}\beta/\eps)^{O(n^k)},$$ so we can index the elements of both sets with elements of $[(n^{k}\beta/\eps)^{O(n^k)}]$. For any two indices $i,j\in [(n^{k}\beta/\eps)^{O(n^k)}],$ we define the observable $\Delta H_{i,j}=H_i-H_j.$ First, we bound the number of copies needed to estimate the expected values of all the observables $\Delta H_{i,j}$ in an unknown state $\rho$.

\begin{lemma}\label{lem:observablesfornetlearning}
    Let $\rho$ be an $n$-qubit state, and let $\eps',\tilde \eps,\delta>0.$ Then, with $O(3^kn^{2k}k\log(n/\delta)/\tilde \eps^2)$ single copies of $\rho$ one can obtain estimates $ \Delta H_{i,j,\rho}'$ such that, with probability $\geq 1-\delta,$
    $$|\Delta H_{i,j,\rho}'-\Tr[\rho\Delta H_{i,j}]|\leq \tilde\eps$$ holds simultaneously for every pair of Hamiltonians $H_i,H_j$ belonging to $\mathcal{H}_{ \eps',n,k,\beta}$. The classical post-processing time is $(n^{k}\beta/\eps')^{O(n^k)}/\widetilde\eps^2.$
\end{lemma}
\begin{proof}
    By the classical shadow estimation protocol of \cref{theo:Shadows}, with $O(3^kn^{2k}k\log(n/\delta)/\tilde \eps^2)$ many copies of $\rho$ one can obtain estimates $\rho_P'$ such that, with probability $\geq 1-\delta$, satisfy $$|2^n\rho_P'-2^n\rho_P|= \frac{\widetilde \eps}{200n^k}$$ for every $|P|\leq k.$ We define $\Delta H_{i,j,\rho}'=\sum_{|P|\leq k}((h_i)_P-(h_j)_P)2^n\rho_P'.$ Then, by Plancherel's identity and \cref{eq:numberofklocalPaulis}, we get
    \begin{align*}
        |\Delta H_{i,j,\rho}'-\Tr[\Delta H_{i,j}\rho]|\leq \sum_{|P|\leq k}|(h_i)_P-(h_j)_P|\cdot 2^n|\rho_P-\rho_P'|\leq 200n^k \max_P|2^n\rho_P'-2^n\rho_P|=\tilde\eps.
    \end{align*}
    The classical post-processing time bound to obtain the estimates $\rho'_P$ is $O(3^kn^{3k}k\log(n)/\tilde \eps^2)$, coming from \cref{theo:Shadows}. Once we have the estimates $\rho'_P,$ by \cref{eq:numberofklocalPaulis}, it takes $O(n^k)$ time to compute each $\Delta H'_{i,j}$. Hence, the total post-processing time is $$O((3^kn^{3k}k\log(n)/\tilde \eps^2)+n^k|\mathcal H_{\eps,k,n,\beta}|^2)=O((3^kn^{3k}k\log(n)/\tilde \eps^2))+n^k(n^{k}\beta/\eps')^{O(n^k)}=(n^{k}\beta/\eps')^{O(n^k)}/\widetilde\eps^2 \, .$$ % \mnote{I don't yet fully get the classical post-processing time. Of course, we get the contribution from classical shadows, but then we also need to use those estimates to compute the $\Delta H_{i,j,\rho}'$ That should be $O(\sum_{|P|\leq k} 1)$ many operations for each of the $O(|\mathcal{H}_{\eps',n,k,\beta}|^2)$ pairs, right? However, why should those complexities multiply? I would expect that they just add up. That is, I would expect a total classical post-processing time of $O(3^kn^{\textcolor{red}{4}k}k\log(n)/\tilde \eps^2) + O(|\mathcal{H}_{\eps',n,k,\beta}|^2\cdot (\sum_{|P|\leq k} 1))$.}
\end{proof}

Now, we are ready to present our Gibbs state learning protocol.

\begin{algorithm}[H]
\caption{Gibbs state learning}\label{alg:quantum_data3}

\begin{algorithmic}[1]
\Require $\delta,\eps\in (0,1)$; $O(3^kn^{2k}k\log(n)(\max\{\beta,1\})^2/ \eps^4)$ single copies of $\rho$. Set $\eps'=\frac{\eps^2}{100\max\{\beta,1\}n^k}$.
\State Obtain $(\eps^2/(\max\{\beta,1\}))$-estimates $\Delta H'_{i,j,\rho}$ of $\Tr(\Delta H_{i,j}\rho)$ with probability $\ge 1-\delta$, for pairs $H_{i}$, $H_j$ belonging to $\mathcal H_{\eps',n,k,\beta}$ via the protocol of \cref{lem:observablesfornetlearning}. 
\State Output $\rho'\in \mathcal{S}_{\eps,n.k,\beta}$, where $$\rho'=\text{argmin}_{\tau\in\mathcal{S}_{\eps',n,k,\beta}}\{\max_{i,j}\{|\Delta H_{i,j,\rho}'-\Tr[\Delta H_{i,j}\tau]|\}\}.$$
\end{algorithmic}
\end{algorithm}

\begin{theorem}\label{theo:GibbsStateLearning}
    Let $\rho$ be the Gibbs state at inverse temperature $\beta$ of an $n$-qubit and $k$-local Hamiltonian $H$ with $|h_P|\leq 1$ for every $P.$ Let $\delta,\eps\in (0,1)$. Then, from $O(3^kn^{2k}k\log(n/\delta)(\max\{\beta,1\})^2/ \eps^4)$ single copies of $\rho$, \Cref{alg:quantum_data3} obtains $\rho'\in\mathcal S_{\eps',n,k,\beta}$ such that $\norm{\rho'-\rho}_{\tr}\leq \eps$ with probability $\ge 1-\delta$. The classical post-processing time of the protocol is $(n^{k}\max\{\beta,1\}/\eps)^{O(n^k)}.$
\end{theorem}
\begin{proof}

    We first show correctness, and then perform a complexity analysis. 

    \textbf{Correctness analysis.}  By \cref{lem:net}, there is $\rho''\in\mathcal{S}_{\eps',n,k,\beta}$ such that $\norm{\rho-\rho''}_{\tr}\leq  \frac{\eps^2}{100\max\{\beta,1\}n^k}\le \eps$. In particular, 
    \begin{align}
        \max_{i,j}\{|\Delta H_{i,j,\rho}'-\Tr[\Delta H_{i,j}\rho'']|\}&= \max_{i,j}\{|(\Delta H_{i,j,\rho}'-\Tr[\Delta H_{i,j}\rho])-\Tr[\Delta H_{i,j}(\rho''-\rho)]|\}\nonumber\\
        &\leq \frac{\eps^2}{\max\{\beta,1\}}+\max_{i,j}\norm{\Delta H_{i,j}}_{\op}\norm{\rho''-\rho}_{\tr}
        \\
        &\leq \frac{\eps^2}{\max\{\beta,1\}}+200n^k\cdot \frac{\eps^2}{100\max\{\beta,1\}n^k}\leq 3\frac{\eps^2}{\beta},\label{eq:auxilary}
    \end{align}
    where in the second line we have used the guarantees of \cref{theo:Shadows}, and in the third line we have used that $\norm{\Delta H_{i,j}}_{\op}\leq 200n^k$ because of \cref{eq:numberofklocalPaulis}. Thus, by definition of $\rho'$, we also have 
    \begin{equation}\label{eq:auxiliary2}
        \max_{i,j}\{|\Delta H_{i,j,\rho}'-\Tr[\Delta H_{i,j}\rho']|\}\leq 3\frac{\eps^2}{\beta}.
    \end{equation}
    
\noindent Now, we are ready to upper bound the trace distance between $\rho'$ and $\rho$. By the triangle inequality we have that 
    \begin{align*}
        \norm{\rho-\rho'}_{\tr}\leq \norm{\rho-\rho''}_{\tr}+\norm{\rho'-\rho''}_{\tr}\leq \eps+\norm{\rho'-\rho''}_{\tr}.
    \end{align*}
    By \cref{lem:dSKLGibbsStates}, Equation \eqref{eq:trnormforGibbs0}, we further have that 
    \begin{align*}
        \norm{\rho'-\rho''}_{\tr}&\leq \sqrt{2\beta \Tr[\Delta H_{l_1,l_0}(\rho'-\rho'')]},
    \end{align*}
    where $l_0$, resp. $l_1$, is the label of the Hamiltonian $H_{l_0}$, resp. $H_{l_1}$, corresponding to state $\rho'$, resp. $\rho''$. Next, we apply \cref{eq:auxilary} and \cref{eq:auxiliary2} and get 
    \begin{equation*}
        \norm{\rho-\rho'}_{\tr}\leq \eps+\sqrt{4\times 3\eps^2}\le 5\eps.
    \end{equation*}
The bound claimed in the statement of the theorem follows up to constant rescaling.
    
    \textbf{Complexity analysis.} The complexities follow from applying \cref{lem:observablesfornetlearning} with $$\eps'=\eps^2/(100\max\{\beta,1\}n^k)\qquad \text{ and } \qquad \tilde\eps=\eps^2/\max\{\beta,1\}.$$
\end{proof}

\subsection{Certifying Gibbs states}
In this section we propose a fully-efficient protocol to certify Gibbs states, i.e., an algorithm whose sample-complexity and time-complexity are both at most polynomial in all relevant parameters. 
\begin{theorem}\label{theo:certifyingGibbsStates}
    Let $\rho$ and $\rho_0$ be the Gibbs states at inverse temperature $\beta$ of $n$-qubit and $k$-local Hamiltonians $H$ and $H_0$ with $|h_P|,|(h_0)_{P}|\leq 1$ for every $P,$ respectively. Assume that $H_0$ is known. Let $\delta,\eps\in (0,1)$. Then, \Cref{alg:quantum_data4} decides, with success probability $\geq 1-\delta$, whether $\norm{\rho-\rho_0}_{\tr}\leq \eps^2/(400\beta n^k)$ or $\norm{\rho-\rho_0}_{\tr}\geq 2\eps$ with $$O
    \left(\frac{\beta^2n^{2k}3^kk\log(n/\delta)}{\eps^4}\right)$$  single copies of $\rho$ and $\rho_0$. Moreover, the protocol only requires Pauli measurements, and a classical post-processing time of order $O
    \left(\beta^2n^{3k}3^kk\log(n/\delta)/\eps^4\right)$. The same conclusion holds if $\rho$ and $\rho_0$ are both unknown and we are given copy access to both.\footnote{As for certain regime it happens that $\eps^2/(400\beta n^k)\geq  2\eps$, it may seem that the \emph{far} and \emph{close} can overlap, and thus that the testing task is not well-defined. However, this is not the case, because for that regime of parameters the \emph{far} hypothesis cannot occur. Indeed, by \cref{lem:dSKLGibbsStates} we have that $\norm{\rho(\beta)-\rho_0(\beta)}_{\tr}\leq \sqrt{2\beta \norm{H-H_0}_{\op}\norm{\rho(\beta)-\rho(\beta)}_{\tr}}$. Then, as $\norm{H-H_0}_{\tr}\leq 200 n^k$, because $|h_P|,|(h_0)_P|\leq 1$, we have that $\norm{\rho(\beta)-\rho_0(\beta)}_{\tr}\leq  400\beta n^k$. For the parameters such that $\eps^2/(400\beta n^k)\geq  2\eps$ we then have that $\norm{\rho(\beta)-\rho_0(\beta)}_{\tr}\leq \eps/2.$}
\end{theorem}

\noindent Since the situation where $\rho$ and $\rho_0$ are both unknown is strictly harder than the case of a known $\rho_0$, we only treat the former. 

\begin{algorithm}[H]
\caption{Gibbs state certification}\label{alg:quantum_data4}
\begin{algorithmic}[1]
\Require $O
    \left(\frac{\beta^2n^{2k}3^kk\log(n/\delta)}{\eps^4}\right)$ single copies of $\rho$, $\delta,\eps\in (0,1)$.
\State Obtain estimates $\rho'_P$ and $(\rho_0)'_P$ such that, with probability $\ge 1-\delta$ via the classical shadow tomography protocol of \cref{theo:Shadows}, such that 
\begin{equation}\label{eq:shadowssucceed}
        2^n|\rho_P-\rho_P'|,2^n|(\rho_0)_P-(\rho_0)_P'|\leq \frac{\eps^2}{800\beta n^k},
    \end{equation} 
    for every $|P|\le k$.
\If{there is $|P|\le k$ such that $2^n|\rho_P'-(\rho_0)_{P}'|\geq 3\eps^2/(400\beta n^k)$} \State output ``FAR''.
\Else \State output ``CLOSE''.
\EndIf
\end{algorithmic}
\end{algorithm}

\begin{proof}
    We first show correctness, and then perform a complexity analysis.

    \textbf{Correctness analysis.} Assume that \cref{eq:shadowssucceed} holds. If $\norm{\rho-\rho_0}_{\tr}\leq \eps^2/(400\beta n^k),$ then $$2^n|\rho_P'-(\rho_0)'_P|\leq 2^n|\rho_P'-\rho_P|+2^n|\rho_P-(\rho_0)_P|+2^n|(\rho_0)_P-(\rho_0)_P'|\leq 2\frac{\eps^2}{400\beta n^k},$$
    so we output ``CLOSE'', as desired. 

    On the other hand, assume that $\norm{\rho-\rho_0}_{\tr}\geq 2\eps.$ Then, by \cref{lem:dSKLGibbsStates}
    \begin{align*}
        4\eps^2\leq 400\beta n^k\max_{|P|\leq k}2^n|\rho_P-(\rho_0)_P|.
    \end{align*}
    Now, by \cref{eq:shadowssucceed} we have that 
    \begin{align*}
        3\eps^2\leq 400\beta n^k\max_{|P|\leq k}2^n|\rho_P'-(\rho_0)_P'|.
    \end{align*}
    Hence, there is $|P|\leq k$ such that $2^n|\rho_P'-(\rho_0)_P'|\geq 3\eps^2/(400\beta n^k)$, as desired.

    \textbf{Complexity analysis.} The complexity analysis follows from applying the classical shadow tomography protocol of \cref{theo:Shadows} with error parameter $\eps^2/(800\beta n^k).$
\end{proof}

\bibliographystyle{alphaurl}
\bibliography{Bibliography}

\newcommand{\etalchar}[1]{$^{#1}$}
\begin{thebibliography}{RSFOW24}

\bibitem[AAKS21]{anshu2021sample}
Anurag Anshu, Srinivasan Arunachalam, Tomotaka Kuwahara, and Mehdi Soleimanifar.
\newblock Sample-efficient learning of interacting quantum systems.
\newblock {\em Nature Physics}, 17(8):931--935, 2021.

\bibitem[ACQ22]{aharonov2022quantum}
Dorit Aharonov, Jordan Cotler, and Xiao-Liang Qi.
\newblock Quantum algorithmic measurement.
\newblock {\em Nature Communications}, 13(1):1--9, 2022.

\bibitem[ADE25]{arunachalam2025testing}
Srinivasan Arunachalam, Arkopal Dutt, and Francisco Escudero{ }Guti{\'e}rrez.
\newblock Testing and learning structured quantum {Hamiltonians}.
\newblock In {\em Proceedings of the 57th Annual ACM Symposium on Theory of Computing}, pages 1263--1270, 2025.

\bibitem[ADEG24]{arunachalam2024testing}
Srinivasan Arunachalam, Arkopal Dutt, and Francisco Escudero~Gutiérrez.
\newblock Testing and learning structured quantum {H}amiltonians, 2024.
\newblock \href {https://arxiv.org/abs/2411.00082} {\path{arXiv:2411.00082}}.

\bibitem[Ans22]{anshu2022some}
Anurag Anshu.
\newblock Some recent progress in learning theory: The quantum side.
\newblock {\em Harvard Data Science Review}, 4(1), 2022.

\bibitem[BCG{\etalchar{+}}24]{bravyi2024high}
Sergey Bravyi, Andrew~W Cross, Jay~M Gambetta, Dmitri Maslov, Patrick Rall, and Theodore~J Yoder.
\newblock High-threshold and low-overhead fault-tolerant quantum memory.
\newblock {\em Nature}, 627(8005):778--782, 2024.

\bibitem[BCO24]{bluhm2024hamiltonian}
Andreas Bluhm, Matthias~C Caro, and Aadil Oufkir.
\newblock Hamiltonian property testing.
\newblock 2024.
\newblock \href {https://arxiv.org/abs/2403.02968} {\path{arXiv:2403.02968}}.

\bibitem[BH17]{bravyi2017complexity}
Sergey Bravyi and Matthew Hastings.
\newblock On complexity of the quantum {Ising} model.
\newblock {\em Communications in Mathematical Physics}, 349(1):1--45, 2017.

\bibitem[BLMT24]{bakshi2023learning}
Ainesh Bakshi, Allen Liu, Ankur Moitra, and Ewin Tang.
\newblock Learning quantum {H}amiltonians at any temperature in polynomial time.
\newblock In {\em Proceedings of the 56th Annual ACM Symposium on Theory of Computing}, pages 1470--1477, 2024.

\bibitem[Bon70]{bonami1970etude}
Aline Bonami.
\newblock Etude des coefficients de {Fourier} des fonctions de {Lp} (g).
\newblock In {\em Annales de l'institut Fourier}, volume~20, pages 335--402, 1970.

\bibitem[BOW19]{buadescu2019quantum}
Costin B{\u{a}}descu, Ryan O'Donnell, and John Wright.
\newblock Quantum state certification.
\newblock In {\em Proceedings of the 51st Annual ACM SIGACT Symposium on Theory of Computing}, pages 503--514, 2019.

\bibitem[CAN25]{chen2025learning}
Chi-Fang Chen, Anurag Anshu, and Quynh~T Nguyen.
\newblock Learning quantum {Gibbs} states locally and efficiently.
\newblock 2025.
\newblock \href {https://arxiv.org/abs/2504.02706} {\path{arXiv:2504.02706}}.

\bibitem[Car24]{caro2023learning}
Matthias~C Caro.
\newblock Learning quantum processes and {Hamiltonians} via the {Pauli} transfer matrix.
\newblock {\em ACM Transactions on Quantum Computing}, 5(2):1--53, 2024.

\bibitem[CM16]{cubitt2016complexity}
Toby Cubitt and Ashley Montanaro.
\newblock Complexity classification of local {Hamiltonian} problems.
\newblock {\em SIAM Journal on Computing}, 45(2):268--316, 2016.

\bibitem[CMP18]{cubitt2018universal}
Toby~S Cubitt, Ashley Montanaro, and Stephen Piddock.
\newblock Universal quantum {Hamiltonians}.
\newblock {\em Proceedings of the National Academy of Sciences}, 115(38):9497--9502, 2018.

\bibitem[CST{\etalchar{+}}21]{childs2021theory}
Andrew~M Childs, Yuan Su, Minh~C Tran, Nathan Wiebe, and Shuchen Zhu.
\newblock Theory of {T}rotter error with commutator scaling.
\newblock {\em Physical Review X}, 11(1):011020, 2021.

\bibitem[CW25]{castaneda2023hamiltonian}
Juan Castaneda and Nathan Wiebe.
\newblock Hamiltonian learning via shadow tomography of pseudo-{Choi} states.
\newblock {\em Quantum}, 9:1700, 2025.

\bibitem[DAC{\etalchar{+}}10]{dutta2010quantum}
Amit Dutta, Gabriel Aeppli, Bikas~K Chakrabarti, Uma Divakaran, Thomas~F Rosenbaum, and Diptiman Sen.
\newblock Quantum phase transitions in transverse field spin models: from statistical physics to quantum information.
\newblock 2010.
\newblock \href {https://arxiv.org/abs/1012.0653} {\path{arXiv:1012.0653}}.

\bibitem[DDK19]{daskalakis2019testing}
Constantinos Daskalakis, Nishanth Dikkala, and Gautam Kamath.
\newblock Testing {Ising} models.
\newblock {\em IEEE Transactions on Information Theory}, 65(11):6829--6852, 2019.

\bibitem[DOS24]{Dutkiewicz.2023}
Alicja Dutkiewicz, Thomas~E O'Brien, and Thomas Schuster.
\newblock The advantage of quantum control in many-body {Hamiltonian} learning.
\newblock {\em Quantum}, 8:1537, 2024.

\bibitem[dSLCP11]{Silva2011Practical}
Marcus~P da~Silva, Olivier Landon-Cardinal, and David Poulin.
\newblock Practical characterization of quantum devices without tomography.
\newblock {\em Physical Review Letters}, 107(21):210404, 2011.

\bibitem[FRF24]{fanizza2024efficient}
Marco Fanizza, Cambyse Rouz{\'e}, and Daniel~Stilck Fran{\c{c}}a.
\newblock Efficient {Hamiltonian}, structure and trace distance learning of {Gaussian} states.
\newblock 2024.
\newblock \href {https://arxiv.org/abs/2411.03163} {\path{arXiv:2411.03163}}.

\bibitem[GCC24]{Gu2022Practical}
Andi Gu, Lukasz Cincio, and Patrick~J Coles.
\newblock Practical {Hamiltonian} learning with unitary dynamics and {Gibbs} states.
\newblock {\em Nature Communications}, 15(1), 2024.

\bibitem[GJW{\etalchar{+}}25]{gao2025quantum}
Minbo Gao, Zhengfeng Ji, Qisheng Wang, Wenjun Yu, and Qi~Zhao.
\newblock Quantum {H}amiltonian certification.
\newblock 2025.
\newblock \href {https://arxiv.org/abs/2505.13217} {\path{arXiv:2505.13217}}.

\bibitem[HBCP15]{holzapfel2015scalable}
M~Holz{\"a}pfel, T~Baumgratz, M~Cramer, and Martin~B Plenio.
\newblock Scalable reconstruction of unitary processes and {Hamiltonians}.
\newblock {\em Physical Review A}, 91(4):042129, 2015.

\bibitem[HHJ{\etalchar{+}}17]{haah2017sample}
Jeongwan Haah, Aram~W Harrow, Zhengfeng Ji, Xiaodi Wu, and Nengkun Yu.
\newblock Sample-optimal tomography of quantum states.
\newblock {\em IEEE Transactions on Information Theory}, 63(9):5628--5641, 2017.

\bibitem[HKP20]{huang2020predicting}
Hsin-Yuan Huang, Richard Kueng, and John Preskill.
\newblock Predicting many properties of a quantum system from very few measurements.
\newblock {\em Nature Physics}, 16(10):1050--1057, 2020.

\bibitem[HKT22]{haah2022optimal}
Jeongwan Haah, Robin Kothari, and Ewin Tang.
\newblock Optimal learning of quantum {H}amiltonians from high-temperature {Gibbs} states.
\newblock In {\em 2022 IEEE 63rd Annual Symposium on Foundations of Computer Science (FOCS)}, pages 135--146. IEEE, 2022.

\bibitem[HMG{\etalchar{+}}25]{hu2025ansatz}
Hong-Ye Hu, Muzhou Ma, Weiyuan Gong, Qi~Ye, Yu~Tong, Steven~T Flammia, and Susanne~F Yelin.
\newblock Ansatz-free {Hamiltonian} learning with {Heisenberg}-limited scaling.
\newblock 2025.
\newblock \href {https://arxiv.org/abs/2502.11900} {\path{arXiv:2502.11900}}.

\bibitem[HTFS23]{huang2023heisenberg}
Hsin-Yuan Huang, Yu~Tong, Di~Fang, and Yuan Su.
\newblock Learning many-body {Hamiltonians} with {Heisenberg}-limited scaling.
\newblock {\em Physical Review Letters}, 130(20):200403, 2023.

\bibitem[Isi25]{ising1925beitrag}
Ernst Ising.
\newblock Beitrag zur {Theorie} des {Ferromagnetismus}.
\newblock {\em Zeitschrift f{\"u}r Physik}, 31(1):253--258, 1925.

\bibitem[KKR06]{kempe2006complexity}
Julia Kempe, Alexei Kitaev, and Oded Regev.
\newblock The complexity of the local {Hamiltonian} problem.
\newblock {\em SIAM Journal on Computing}, 35(5):1070--1097, 2006.

\bibitem[KL25]{kallaugher2025hamiltonian}
John Kallaugher and Daniel Liang.
\newblock Hamiltonian locality testing via trotterized postselection.
\newblock 2025.
\newblock \href {https://arxiv.org/abs/2505.06478} {\path{arXiv:2505.06478}}.

\bibitem[LSG{\etalchar{+}}25]{liu2025robust}
Hua-Liang Liu, Hao Su, Si-Qiu Gong, Yi-Chao Gu, Hao-Yang Tang, Meng-Hao Jia, Qian Wei, Yukun Song, Dongzhou Wang, Mingyang Zheng, et~al.
\newblock Robust quantum computational advantage with programmable 3050-photon {Gaussian} boson sampling.
\newblock 2025.
\newblock \href {https://arxiv.org/abs/2508.09092} {\path{arXiv:2508.09092}}.

\bibitem[LTG{\etalchar{+}}24]{li2023heisenberglimited}
Haoya Li, Yu~Tong, Tuvia Gefen, Hongkang Ni, and Lexing Ying.
\newblock Heisenberg-limited {Hamiltonian} learning for interacting bosons.
\newblock {\em npj Quantum Information}, 10(1):83, 2024.

\bibitem[LW22]{laborde2022quantum}
Margarite~L LaBorde and Mark~M Wilde.
\newblock Quantum algorithms for testing {Hamiltonian} symmetry.
\newblock {\em Physical Review Letters}, 129(16):160503, 2022.

\bibitem[MBC{\etalchar{+}}23]{möbus2023dissipationenabled}
Tim Möbus, Andreas Bluhm, Matthias~C Caro, Albert~H Werner, and Cambyse Rouzé.
\newblock Dissipation-enabled bosonic {Hamiltonian} learning via new information-propagation bounds, 2023.
\newblock \href {https://arxiv.org/abs/2307.15026} {\path{arXiv:2307.15026}}.

\bibitem[MFPT24]{ma2024learning}
Muzhou Ma, Steven~T Flammia, John Preskill, and Yu~Tong.
\newblock Learning $ k $-body {Hamiltonians} via compressed sensing.
\newblock 2024.
\newblock \href {https://arxiv.org/abs/2410.18928} {\path{arXiv:2410.18928}}.

\bibitem[MM12]{mccoy2012importance}
Barry~M McCoy and Jean-Marie Maillard.
\newblock The importance of the {Ising} model.
\newblock {\em Progress of Theoretical Physics}, 127(5):791--817, 2012.

\bibitem[MO08]{montanaro2008quantum}
Ashley Montanaro and Tobias~J Osborne.
\newblock Quantum {B}oolean functions.
\newblock {\em Chicago Journal of Theoretical Computer Science}, pages 1--45, 2008.

\bibitem[OT08]{oliveira2005complexity}
Roberto Oliveira and Barbara~M Terhal.
\newblock The complexity of quantum spin systems on a two-dimensional square lattice.
\newblock {\em Quantum Information \& Computation}, 8(10):900--924, 2008.

\bibitem[OW15]{odonnell2015quantum}
Ryan O'Donnell and John Wright.
\newblock Quantum spectrum testing.
\newblock In {\em Proceedings of the 47th annual ACM symposium on Theory of computing (STOC)}, pages 529--538, 2015.

\bibitem[OW16]{o2016efficient}
Ryan O'Donnell and John Wright.
\newblock Efficient quantum tomography.
\newblock In {\em Proceedings of the Forty-Eighth Annual ACM Symposium on Theory of Computing}, page 899–912, 2016.

\bibitem[RF24]{rouze2023learning}
Cambyse Rouz{\'e} and Daniel~Stilck Fran{\c{c}}a.
\newblock Learning quantum many-body systems from a few copies.
\newblock {\em Quantum}, 8:1319, 2024.

\bibitem[RSFOW24]{onorati2023efficient}
Cambyse Rouz{\'e}, Daniel Stilck~Fran{\c{c}}a, Emilio Onorati, and James~D Watson.
\newblock Efficient learning of ground and thermal states within phases of matter.
\newblock {\em Nature Communications}, 15(1):7755, 2024.

\bibitem[SFMD{\etalchar{+}}24]{franca2024efficient}
Daniel Stilck~Fran{\c c}a, Liubov~A Markovich, V~V Dobrovitski, Albert~H Werner, and Johannes Borregaard.
\newblock Efficient and robust estimation of many-qubit {Hamiltonians}.
\newblock {\em Nature Communications}, 15:311, 2024.

\bibitem[SIC12]{suzuki2012quantum}
Sei Suzuki, Jun-ichi Inoue, and Bikas~K Chakrabarti.
\newblock {\em Quantum {Ising} phases and transitions in transverse {Ising} models}, volume 862.
\newblock Springer, 2012.

\bibitem[ST25]{sinha2025improvedhamiltonianlearningsparsity}
Savar~D Sinha and Yu~Tong.
\newblock Improved {Hamiltonian} learning and sparsity testing through {Bell} sampling, 2025.
\newblock \href {https://arxiv.org/abs/2509.07937} {\path{arXiv:2509.07937}}.

\bibitem[SW12]{santhanam2012information}
Narayana~P Santhanam and Martin~J Wainwright.
\newblock Information-theoretic limits of selecting binary graphical models in high dimensions.
\newblock {\em IEEE Transactions on Information Theory}, 58(7):4117--4134, 2012.

\bibitem[SY23]{she2022unitary}
Adrian She and Henry Yuen.
\newblock {Unitary property testing lower bounds by polynomials}.
\newblock In {\em 14th Innovations in Theoretical Computer Science Conference (ITCS 2023)}, volume 251, pages 96:1--96:17, 2023.

\bibitem[YSHY23]{yu2023robust}
Wenjun Yu, Jinzhao Sun, Zeyao Han, and Xiao Yuan.
\newblock Robust and efficient {Hamiltonian} learning.
\newblock {\em Quantum}, 7:1045, 2023.

\bibitem[Zha25]{zhao2025learning}
Andrew Zhao.
\newblock Learning the structure of any {Hamiltonian} from minimal assumptions.
\newblock In {\em Proceedings of the 57th Annual ACM Symposium on Theory of Computing}, pages 1201--1211, 2025.

\bibitem[ZYLB22]{Zubida2021Optimal}
Assaf Zubida, Elad Yitzhaki, Netanel Lindner, and Eyal Bairey.
\newblock Optimal short-time measurements for {Hamiltonian} learning.
\newblock {\em Bulletin of the American Physical Society}, 67, 2022.

\end{thebibliography}
\end{document}